\documentclass[letter, 10pt, conference]{ieeeconf}      % Use this line for a4 paper

\IEEEoverridecommandlockouts                              % This command is only needed if 
\usepackage{amsmath}
\usepackage{amssymb}
\usepackage{amsfonts}
\usepackage{graphicx}
\usepackage{enumerate}
\usepackage{mathrsfs}
\usepackage{float}
\usepackage{cite}
\usepackage{tikz}
\usetikzlibrary{matrix,arrows,decorations.pathmorphing}
\usepackage{verbatim}
\usepackage{mathtools}
\mathtoolsset{showonlyrefs}
\usepackage{caption}
\usepackage{subcaption}
\usepackage{soul}
\setstcolor{magenta}
\usetikzlibrary{arrows}
\usepackage{tabularx}

\newcommand{\hma}[1]{\textcolor{black}{#1}}
\newcommand{\hjk}[1]{\textcolor{black}{#1}}

\graphicspath{{./figs/}}

\title{\LARGE \bf
Markov Chains and Random Walks with Memory on Hypergraphs: A Tensor-Based Approach$^{*}$ }

\author{Shaoxuan Cui$^{1,4}$, Lingfei Wang$^{2}$, Hildeberto Jardón-Kojakhmetov$^{1}$, Karl Henrik Johansson$^{2}$ and Ming Cao$^{3}$ % stops a space
\thanks{$^{1}$ S. Cui, and H. Jard\'on-Kojakhmetov are with the Bernoulli Institute for Mathematics, Computer Science and Artificial Intelligence, University of Groningen, Groningen, 9747 AG Netherlands {\tt\small \{s.cui, h.jardon.kojakhmetov\}@rug.nl}}
\thanks{$^{2}$ L. Wang and K. Johansson are with the Division of Decision and Control Systems, the School of Electrical Engineering and Computer Science, KTH Royal Institute of Technology {\tt\small \{lingfei, kallej\}@kth.se}}
\thanks{$^{3}$ M. Cao is with the Engineering and Technology institute Groningen, University of Groningen, Groningen, 9747 AG Netherlands {\tt\small m.cao@rug.nl}}
\thanks{$^{*}$ The work was supported by the Netherlands Organization for Scientific Research (NWO-Vici-19902), the Knut and Alice Wallenberg Foundation
(Wallenberg Scholar Grant), and the Swedish Research Council (Distinguished
Professor Grant 2017-01078). }
}

\begin{document}

\maketitle

\thispagestyle{empty}
\pagestyle{empty}

\newtheorem{remark}{Remark}
\newtheorem{lemma}{Lemma}
\newtheorem{thm}{Theorem}
\newtheorem{example}{Example}
\newtheorem{definition}{Definition}
\newtheorem{prop}{Proposition}

%%%%%%%%%%%%%%%%%%%%%%%%%%%%%%%%%%%%%%%%%%%%%%%%%%%%%%%%%%%%%%%%%%%%%%%%%%%%%%%%
\begin{abstract}
Many complex systems exhibit interactions that depend not only on pairwise connections, 
but also group structures and memory effects. 
To capture such effects, we develop a unified tensor framework for modeling higher-order Markov chains with memory. 
Our formulation introduces an even-order paired tensor that links folded and unfolded dynamics 
and characterizes their steady states and convergence. We further show that a Markov chain with memory can be approximated by a low-dimensional nonlinear tensor-based system and then provide a full system analysis.
As an application, we define random walks on hypergraphs where memory naturally arises from the hyperedge structure, 
providing new tools for analyzing higher-order networks with time-dependent effects.
\end{abstract}


%%%%%%%%%%%%%%%%%%%%%%%%%%%%%%%%%%%%%%%%%%%%%%%%%%%%%%%%%%%%%%%%%%%%%%%%%%%%%%%%
\section{Introduction}

\hjk{
Markov chains are a fundamental tool for modeling stochastic processes in networks \cite{norris1998markov,levin2017markov}, but classical formulations assume memoryless, pairwise interactions.
Many real-world systems violate both assumptions, for instance: biochemical pathways involve simultaneous multi-molecule reactions whose ordering determines downstream behavior \cite{battiston2021physics,alon2019introduction}; coordinated neuronal firing patterns form temporal motifs beyond pairwise description \cite{giusti2016two}; and information diffusion in social networks exhibits both group-wise interactions \cite{cui2025sis,cui2024metzler} and strong temporal correlations \cite{holme2013temporal,fischer2020visual}.
Neglecting such higher-order and non-Markovian structure can lead to misleading characterizations of diffusion, ranking, or control processes \cite{lambiotte2019networks}.
}

\hjk{
Two independent lines address these challenges.
\emph{Markov chains with memory} \cite{wu2017markov} allow transitions to depend on multiple past states, but existing tensor formulations are algebraically cumbersome, with explicit unfolding available only for memory depth two.
\emph{Hypergraph-based random walks} \cite{carletti2020random,hayashi2020hypergraph} capture group interactions but remain memoryless.
No unified framework combines memory with hypergraph structure.
Meanwhile, tensor methods have proven effective for dynamics on hypergraphs \cite{cui2025sis,cui2024metzler,cui2025higher,cui2025analysis} and for multilinear control systems \cite{chen2021multilinear,wang2024algebraic}, motivating a tensor-based unification.
}

In this work, we introduce a \emph{tensor-unfolding framework that unifies higher-order Markov chains with memory and random walks on hypergraphs}. 
The central idea is to represent each memory-driven transition as motion through a directed hyperedge: the ordered tail encodes the sequence of past states, and the head specifies the next state. 
This construction defines a new class of random walks in which the transition probabilities explicitly depend on ordered sequences of past states. By representing memory sequences as ordered hyperedges, the framework retains the generality of higher-order Markov processes while embedding them in the combinatorial setting of hypergraph random walks. 
The outcome is a representation that is both mathematically tractable and intuitively interpretable, highlighting how memory fundamentally shapes diffusion in complex systems.  
\textbf{The main contributions of this paper are threefold:}
%\begin{itemize}
 %   \item We propose a \emph{general tensor-unfolding method} that yields an explicit representation for Markov chains with arbitrary finite memory depth.  
  %  \item We demonstrate that these dynamics can be \emph{approximated by a low-dimensional nonlinear system}, and we analyze its qualitative behavior in detail.  
   % \item We establish a \emph{new class of memory-aware random walks on directed hypergraphs}, providing a natural and physically meaningful model of higher-order diffusion.
%\end{itemize}
    \hma{1. We propose an even-order paired tensor formulation for Markov chains with arbitrary finite memory depth, providing an explicit tensor representation of the unfolded memory process. 
    2. We extend this formulation to continuous-time Markov chains with memory and, under a closure, derive a low-dimensional nonlinear model whose system behaviors are characterized. 
    3. We introduce memory-aware random walks on directed hypergraphs, in which hyperedges encode ordered memory transitions, yielding a natural diffusion model that jointly captures higher-order structure and memory effects.}

Together, these results open a pathway to studying memory effects in network dynamics with tools that are both rigorous and broadly applicable.

\textbf{Notation:} The sets of real and complex numbers are $\mathbb{R}$ and $\mathbb{C}$, respectively. 
The all-ones(zeros) vector is $\mathbf{1}$ ($\mathbf{0}$) and $\|\cdot\|_2$ denotes the Euclidean norm. For any two vectors $a, b \in \mathbb{R}^n$, $a > (<) b$ indicates that $a_i >(<) b_i$, for all $i=1,\ldots,n$. These component-wise comparisons are also valid for matrices or tensors with the same dimension.

%The remainder of the paper is organized as follows.
%Section~\ref{sec:preliminaries} introduces the preliminaries on even-order paired tensors and tensor unfolding.
%Section~\ref{sec:markov_memory} presents the proposed framework for Markov chains with memory and derives theoretical results on convergence and steady states.
%Section~\ref{sec:random_walks} applies this framework to random walks on directed hypergraphs.
%Finally, Section~\ref{sec:conclusion} concludes the paper and outlines future extensions, including non-uniform and multi-agent hypergraph settings.\cite{9778188}

\section{Preliminaries on Tensors and Hypergraphs}
\label{sec:preliminaries}

In this section, we review the basic tensor operations and definitions used throughout the paper,
and describe how hypergraphs can be naturally represented in this framework.

\subsection{General Tensors and Eigenvalues}
A tensor \(T \in \mathbb{C}^{n_1 \times n_2 \times \cdots \times n_k}\) is a multidimensional array. In this paper, all tensors \(T \in \mathbb{R}^{n_1 \times n_2 \times \cdots \times n_k}\) are real except the eigentensor in Definition \ref{def:Uirreducible}. 
The \emph{order} of a tensor is \(k\), representing the number of dimensions, 
where each dimension \(n_i\) (\(i = 1, \ldots, k\)) is referred to as a \emph{mode}. 
If all modes have the same dimension, the tensor is called \emph{cubical} and is denoted as 
\(T \in \mathbb{R}^{n \times n \times \cdots \times n}\). 
A cubical tensor is said to be \emph{supersymmetric} if its entries remain invariant under any permutation of indices.

%A tensor is a multidimensional array that generalizes the notions of scalars (order-0), vectors (order-1), 
%and matrices (order-2). 
%Throughout this paper, we use an order-$m$ tensor ${A} = [a_{i_1 i_2 \cdots i_m}],$
%where each index $i_k$ runs over a finite set of size $n_k$. 
%When all dimensions are equal, i.e., $n_1 = n_2 = \cdots = n_m = n$, 
%the tensor is called \emph{cubical}.
\begin{definition}[Z- and H-eigenpairs {\cite{qi2005eigenvalues}}]
Let ${A} \in \mathbb{R}^{n \times \cdots \times n}$ be a cubical tensor of order $m$. 
For a vector $x \in \mathbb{R}^n$ and a scalar $s\in\mathbb N$, define $x^{[s]} := (x_1^s, x_2^s, \ldots, x_n^s)^\top$. 
A scalar $\lambda \in \mathbb{R}$ (eigenvalue) and a vector $x \in \mathbb{R}^n$ (eigenvector) form an eigenpair of ${A}$ if \hjk{${A} x^{m-1} = \lambda x^{[s]}$,}
% \begin{equation}\label{eq:z_h_eig_def}
% {A} x^{m-1} = \lambda x^{[s]},
% \end{equation}
where
$\big({A} x^{m-1}\big)_i = \sum_{i_2,\ldots,i_m=1}^n A_{i\,i_2 \cdots i_m} x_{i_2} \cdots x_{i_m}.$
Specifically,
 for a \emph{Z-eigenpair}, $s=1$ with the normalization $\|x\|_2 = 1$;
     for an \emph{H-eigenpair}, $s=m-1$ without additional normalization.

\end{definition}

\subsection{Even-Order Paired Tensors and Einstein Product}

We now introduce the notion of an \emph{even-order paired tensor}, which provides a convenient
representation for systems with multiple interacting components and memory structures.

An even-order paired tensor $A \in \mathbb{R}^{I_1 \times J_1 \times \cdots \times I_N \times J_N}$ is a $2N$-th order tensor
whose indices are arranged in $N$ \emph{pairs} $(i_n, j_n)$ for $n = 1,2,\ldots, N$.
For instance, the $(2n-1)$-th mode corresponds to the \emph{$n$-mode row} and the $2n$-th mode
to the \emph{$n$-mode column}.

\begin{definition}[Einstein Product \cite{wang2024algebraic,chen2021multilinear}]
Given two even-order paired tensors
$A \in \mathbb{R}^{I_1 \times J_1 \times \cdots \times I_N \times J_N}$ and
$B \in \mathbb{R}^{J_1 \times K_1 \times \cdots \times J_N \times K_N}$,
their \emph{Einstein product}, denoted by $A \ast B$, is defined as
    $(A \ast B)_{i_1k_1\cdots i_Nk_N} =
    \sum_{j_1=1}^{J_1} \cdots \sum_{j_N=1}^{J_N}
    A_{i_1j_1\cdots i_Nj_N} B_{j_1k_1\cdots j_Nk_N}.$
\end{definition}

The Einstein product generalizes standard matrix multiplication to higher-order settings.
If $Y \in \mathbb{R}^{J_1 \times \cdots \times J_N}$ is a regular $N$-th order tensor,
we can interpret $Y$ as a special even-order paired tensor, in which all mode column dimensions are equal to 1. Then, we define
\[
    (A \ast Y)_{i_1\cdots i_N} =
    \sum_{j_1=1}^{J_1} \cdots \sum_{j_N=1}^{J_N} A_{i_1j_1\cdots i_Nj_N} Y_{j_1\cdots j_N}.
\]

\subsection{Tensor Unfolding}

To simplify computations, it is often convenient to transform a tensor into a matrix or vector. This process is called \emph{Tensor Unfolding}. \hjk{This unfolding relies on the following index map.}
%First of all, we should define a map $\mathrm{ivec}(\cdot)$ that converts a multi-index tensor into a matrix or vector form, which is defined as follows.

\begin{definition}[Index Vectorization Function \cite{wang2024algebraic,chen2021multilinear}]
For a multi-index $i = (i_1, i_2, \ldots, i_N)$ and dimension sizes $I = (I_1, I_2, \ldots, I_N)$,
the mapping $\mathrm{ivec}(i, I)$ flattens the multi-index into a single integer index:
\begin{equation}\label{eq:ivec}
\mathrm{ivec}(i, I) = i_1 + \sum_{k=2}^N (i_k - 1) \prod_{j=1}^{k-1} I_j.
\end{equation}
\end{definition}

Then, the tensor unfolding is defined as follows.

\begin{definition}[Tensor Unfolding \cite{wang2024algebraic,chen2021multilinear}]\label{def:unfolding}
Given an even-order paired tensor $A \in \mathbb{R}^{I_1 \times J_1 \times \cdots \times I_N \times J_N},$
% \[
% A \in \mathbb{R}^{I_1 \times J_1 \times \cdots \times I_N \times J_N},
% \]
%
its \emph{unfolding} is a matrix $\bar{A} = \varphi(A)$ of size $|I| \times |J|$, where
\[
|I| = \prod_{n=1}^N I_n, 
\qquad
|J| = \prod_{n=1}^N J_n,
\]
and with the unfolding map $\varphi$ is explicitly defined by
\begin{equation}\label{eq:unfolding}
A_{i_1 j_1 \cdots i_N j_N}
\quad \xmapsto{\quad \varphi \quad} \quad
\bar{A}_{\mathrm{ivec}(i,I),\,\mathrm{ivec}(j,J)},
\end{equation}
where $i = (i_1, \ldots, i_N)$ and $j = (j_1, \ldots, j_N)$
are the row and column multi-indices, respectively. For simplicity, $\mathrm{ivec}(i,I)=  \mathrm{ivec}(i_1, \ldots, i_N).$
\end{definition}

The Einstein product then simplifies to standard matrix multiplication: \hjk{$\varphi(A \ast B) =\varphi(A) \varphi(B) =\bar{A} \, \bar{B}$.}
% \[
%     \varphi(A \ast B) =\varphi(A) \varphi(B) =\bar{A} \, \bar{B}.
% \]

This unfolding allows us to leverage existing linear algebra techniques for analysis of higher-order time-dependent systems. %The same mapping can be applied to an $N$-th order tensor $Y \in \mathbb{R}^{J_1 \times \cdots \times J_N}$
%by interpreting $Y$ as a special case of an even-order paired tensor with column dimensions equal to one \hjk{This is already mentioned above, try to avoid repetitions}.
%In this case, $\varphi(Y)$ produces a vector $\bar{y}$.

\subsection{U-Eigenvalues and Perron-Frobenius Theorem for Tensors}

We next define another type of higher-order eigenvalues and eigenvectors.

\begin{definition}[U-Eigenvalue  \cite{wang2024algebraic,chen2021multilinear}]
Let $A \in \mathbb{R}^{I_1 \times I_1 \times \cdots \times I_N \times I_N}$ be an even-order paired tensor.
If there exists a non-zero tensor $\Tilde{X} \in \mathbb{C}^{I_1 \times \cdots \times I_N}$ and scalar $\lambda \in \mathbb{C}$
such that
%\begin{equation}
    $A \ast \Tilde{X} = \lambda \Tilde{X},$
%\end{equation}
then $\lambda$ and $\Tilde{X}$ are called the \emph{U-eigenvalue} and \emph{U-eigentensor} of $A$, respectively.
\end{definition}

Consequently, we propose the following.

\begin{definition}[U-Irreducible Tensor]\label{def:Uirreducible}
An even-order paired tensor 
$A \in \mathbb{R}^{I_1 \times I_1 \times \cdots \times I_N \times I_N}$ 
is said to be \emph{U-irreducible} (\emph{U-primitive}) if and only if its unfolded matrix 
$\bar{A} = \varphi(A)$ is \emph{irreducible} (primitive) in the classical sense \cite{FB-LNS}.

\end{definition}

%For nonnegative and irreducible tensors, after applying the unfolding map $\varphi$, the tensor behaves like a standard irreducible nonnegative matrix. Thus, a natural generalization of the Perron-Frobenius theorem holds:

%\hjk{
%The unfolding map $\varphi$ produces a standard nonnegative irreducible matrix for a nonnegative and irreducible tensor. Consequently, the classical Perron–Frobenius theorem applies directly to $\varphi(\mathcal{A})$, yielding a natural extension of Perron–Frobenius theory to the tensor setting:
%}

\begin{lemma}[Perron-Frobenius theorem]
\hjk{By the classical Perron-Frobenius theorem \cite{FB-LNS}, if} $A$ is a nonnegative and U-irreducible even-order paired tensor,
then there exists a unique positive U-eigentensor $X > \mathbf{0}$
associated with the largest real eigenvalue $\lambda > 0$.
Moreover, $\lambda$ is simple and dominates all other eigenvalues in modulus.
\end{lemma}

%The above Lemma is the consequence of the well-known Perron-Frobenius theorem of an irreducible nonnegative matrix \cite{FB-LNS}. This result plays a central role in characterizing stationary distributions in later sections.

\subsection{Hypergraphs and their Tensor Representation}

A \emph{hypergraph} \cite{bick2023higher,gallo1993directed} is a generalization of a graph where each hyperedge can connect
more than two nodes simultaneously.
A hypergraph is called \emph{$k$-uniform} if every hyperedge contains exactly $k$ nodes. 
In this paper, we consider directed $k$-uniform hypergraphs, where each hyperedge has a single head and $(k-1)$ ordered tail nodes.
The structure of such a hypergraph can be encoded by an adjacency tensor
$A \in \mathbb{R}^{n \times n \times \cdots \times n}$ ($k$ modes total), where
\[ \small
    A_{i_1 i_2 \cdots i_k} 
    \begin{cases}
    \neq 0, & \text{if there is a hyperedge from } i_2, \ldots, i_k \text{ to } i_1 \\[4pt]
    =0, & \text{otherwise}.
    \end{cases}
\]
The \emph{tail degree tensor} ${D}$ is a diagonal operator defined over all $(k-1)$-tuples of tail nodes:
 $ {D}_{i_2  \cdots i_k } = 
    \sum_{i_1=1}^n A_{i_1 i_2 \cdots i_k}.$
Using this definition, the normalized adjacency tensor is given by
    $\Tilde{A}_{i_1 i_2 \cdots i_k} =
    \frac{A_{i_1 i_2 \cdots i_k}}
         {{D}_{i_2  \cdots i_k }},$
where each fixed tail $(i_2,\ldots,i_k)$ is normalized individually to ensure
    $\sum_{i_1=1}^n \Tilde{A}_{i_1 i_2 \cdots i_k} = 1.$
This construction guarantees that $\Tilde{A}$ is suitable for modeling random walks on hypergraphs (see Section~\ref{sec:random_walks}).

\medskip
In the next section, we build on these definitions to develop a tensor unfolding framework
for Markov chains with memory.

\section{Higher-Order Markov Chains with Memory}
\label{sec:markov_memory}

\hjk{We now develop a unified tensor-based framework for Markov chains with memory, generalizing the formulation of \cite{wu2017markov} to arbitrary memory depth.}

% We first formalize the notion of a Markov chain with memory, which generalizes the classical first-order Markov chain by allowing the next state to depend on a finite sequence of past states rather than only the most recent one.
% Such models have been studied in various contexts, and a notable tensor-based formulation was introduced in \cite{wu2017markov}, which expresses the transition dynamics compactly using higher-order tensors.
% Building on this foundation, we further develop a unified and more general tensor-based representation.

\subsection{Definition and Basic Properties}

Consider a finite state space $\mathcal{V} = \{1, 2, \ldots, n\}$.
A stochastic process $\{X_t\}_{t\geq 0}$ is said to be a \emph{Markov chain
with memory $m$} if the transition probability satisfies
%\begin{equation*}
%\begin{split}
 $   \Pr(X_{t+1} = i_1 \mid X_t = i_2, X_{t-1} = i_3, \ldots, X_{t-m+2} = i_m) 
= p_{i_1 i_2 \cdots i_m},$
%\end{split}
%\end{equation*}
where
\[
p_{i_1 i_2 \cdots i_m} \geq 0, 
\quad \sum_{i_1=1}^n p_{i_1 i_2 \cdots i_m} = 1, 
\quad \forall (i_2,\ldots,i_m) \in \mathcal{V}^{m-1}.
\]
%The quantity $p_{i_1 i_2 \cdots i_m}$ specifies the probability that the system moves to state $i_1$ at time $t+1$, given that its previous $m-1$ states were $i_2,\ldots,i_m$. 
When $m=2$, this reduces to a standard first-order Markov chain.

\subsection{Memory Process as a First-Order Markov Chain}

To analyze this process, it is common to define an \emph{augmented state vector} 
that encodes the last $m-1$ states:
\[
Y_t = (X_t, X_{t-1}, \ldots, X_{t-m+2}) \in \mathcal{V}^{m-1}.
\]
The evolution of $\{Y_t\}$ is then a first-order Markov chain on a state space
of size $n^{m-1}$. Its transition probability from $Y_t = (i_2,\ldots,i_m)$
to $Y_{t+1} = (i_1,\ldots,i_{m-1})$ is exactly $p_{i_1 i_2 \cdots i_m}$.

The joint distribution of the augmented state vector is then a column vector
$\pi_t \in \mathbb{R}^{n^{m-1}}$ satisfying the standard iteration
$\pi_{t+1} = M \pi_t,$
where $M$ is a transition matrix and contains the information of all probabilities $p$'s.
\hjk{
However, constructing $M$ explicitly requires $n^{m-1}\times n^{m-1}$ entries, and its connection to the original transition probabilities remains implicit \cite{wu2017markov}. This motivates the following tensorial representation.
}

% While conceptually straightforward, constructing $M$ explicitly leads to a state-space blow-up of dimension $n^{m-1}$, which quickly becomes computationally infeasible for large $n$ or $m$ \cite{wu2017markov}.
% This stacking approach can be viewed as a \emph{manual unfolding} of the original memory-dependent process into a high-dimensional standard Markov chain.
% However, the resulting transition matrix $M$ is built purely through index manipulations, and its mathematical connection to the original system remains implicit. This motivates us to propose the next tensorial representation.

\subsection{Tensor Representation}

Instead of explicitly constructing the $n^{m-1}$-dimensional transition matrix,
we directly encode the memory process in a higher-order transition tensor.
Define the order-$m$ transition tensor
$P = [p_{i_1 i_2 \cdots i_m}] \in \mathbb{R}^{n \times n \times \cdots \times n},$
where the first index $i_1$ corresponds to the next state and
$(i_2,\ldots,i_m)$ correspond to the ordered history of length $m-1$. Each entry naturally and directly exhibits the physical meaning of the transition.

For the special case $m=2$, $P$ reduces to a standard $n \times n$ stochastic matrix.
For $m>2$, $P$ compactly captures the transition rules without requiring an explicit expansion.

The dynamics of the memory process can now be written directly in terms of $P$.
Let $\Pi_t$ be the joint probability tensor of the last $m-1$ states,
\[
(\Pi_t)_{i_2 \cdots i_m} = \Pr(X_t = i_2, X_{t-1} = i_3, \ldots, X_{t-m+2} = i_m).
\]
Then the update rule for $\Pi_t$ is given by
\begin{equation}\label{eq:Pi_update}
(\Pi_{t+1})_{i_1 \cdots i_{m-1}} = 
\sum_{i_m=1}^n p_{i_1 i_2 \cdots i_m} (\Pi_t)_{i_2 \cdots i_m}.
\end{equation}

Let $x_t \in \mathbb{R}^n$ denote the state distribution at time $t$, i.e.,
\[
(x_t)_i := \Pr(X_t = i), \qquad \sum_{i=1}^n (x_t)_i = 1.
\]
Since $\Pi_t$ collects the joint distribution of the last $m{-}1$ states, $x_t$ is the following sum: \hjk{$(x_t)_i \;=\; \sum_{i_2,\ldots,i_{m-1}=1}^n (\Pi_t)_{\,i\, i_2 \cdots i_{m-1}}$.}
% \begin{equation}\label{eq:xt_marginal}
% (x_t)_i \;=\; \sum_{i_2,\ldots,i_{m-1}=1}^n (\Pi_t)_{\,i\, i_2 \cdots i_{m-1}}.
% \end{equation}

\subsection{Even-Order Paired Tensor Formulation}

While the transition tensor $P$ compactly represents the memory process,
it is still indexed asymmetrically, with the ``next state'' separated from the
``past states''. To reveal the inherent symmetry and enable structured analysis,
we lift $P$ to an even-order paired tensor $\tilde{P}$ of order $2(m-1)$.

Specifically, let $i = (i_1,\ldots,i_{m-1})$ and $j = (j_1,\ldots,j_{m-1})$
be multi-indices representing the ``head'' and ``tail'' of a transition.
We define
\begin{equation}\label{eq:P_paired}
\tilde{P}_{i_1 j_1 i_2 j_2 \cdots i_{m-1} j_{m-1}} =
\begin{cases}
p_{i_1 i_2 \cdots i_m}, & \text{if } j_1 = i_2, j_2 = i_3, \\[4pt]
& \ldots, j_{m-1} = i_m, \\[4pt]
0, & \text{otherwise}.
\end{cases}
\end{equation}

The Einstein product of $\tilde{P}$ with the joint distribution tensor $\Pi_t$
then yields the updated state:
\begin{equation}\label{eq:pit}
\Pi_{t+1} = \tilde{P} \ast \Pi_t.
\end{equation}
By applying the unfolding map $\varphi(\cdot)$ introduced in Section~\ref{sec:preliminaries},
\eqref{eq:pit} becomes a standard linear iteration \hjk{$\varphi(\Pi_{t+1}) = \varphi(\tilde{P}) \, \varphi(\Pi_t)$.}
% \[
% \varphi(\Pi_{t+1}) = \varphi(\tilde{P}) \, \varphi(\Pi_t).
% \]
Thus, the even-order paired tensor $\tilde{P}$ fully characterizes the Markov chain with memory.
\hjk{\begin{remark}
    In \cite{wu2017markov}, an explicit unfolding was given only for $m=2$. Our formulation provides a constructive procedure for any $m\geq2$, simplifying both analysis and computation.
\end{remark}}

% \begin{remark}
% In \cite{wu2017markov}, an explicit unfolding into matrix-vector form
% was given only for the special case $m=2$. 
% Our even-order paired tensor formulation provides a constructive and explicit
% procedure for \emph{any} memory depth $m \geq 2$.
% This unification simplifies both the theoretical analysis and numerical computation,
% and highlights the underlying structure of the process that was previously hidden
% in index manipulations.
% \end{remark}

\subsection{Stationary Distributions and Convergence}

Since $\varphi(\tilde{P})$ is a nonnegative column-stochastic matrix,
the classical Markov Chain theory can be applied directly.
We have the following result.

\begin{thm}[Convergence behavior]
\label{thm:convergence}
Suppose that $\widetilde{P}$ is the extended transition probability tensor of a Markov chain with memory depth $m-1$. 
Assume that $\widetilde{P}$ is \emph{U-primitive}, i.e., its unfolding 
$\varphi(\widetilde{P})$ is a primitive nonnegative column-stochastic matrix. 
Then, for any generic initial joint memory distribution $\Pi_0 = \Pi_{0,-1,\ldots,-m+1}$ ($0,-1,\ldots,-m+1$ denotes a virtual time instant denoting the initial sequence history), the following statements hold:
\begin{enumerate}
    \item \textbf{Perron-U-eigenvalue:} The tensor $\widetilde{P}$ has a unique dominant U-eigenvalue 
    $\lambda_1(\widetilde{P}) = 1$, which is simple and strictly greater in modulus than all other U-eigenvalues. 
    Its corresponding U-eigentensor $\widetilde{\Pi} > \mathbf{0}$ is strictly positive.

    \item \textbf{Convergence:} The sequence of joint probability mass functions $\{\Pi_t\}$ generated by the update rule \eqref{eq:Pi_update} converges to the unique limit $\widetilde{\Pi}$\hjk{, that is $\lim_{t \to \infty} \Pi_t = \widetilde{\Pi}.$}
  
    \item \textbf{Stationary Distribution:} The stationary distribution $\tilde{x}$ of the original Markov chain with memory $m$ exists and is given by the marginal sum of $\widetilde{\Pi}$:
    \begin{equation}\label{eq:xsum}
        \tilde{x}_i = \sum_{i_2, \ldots, i_{m-1}} \widetilde{\Pi}_{i\, i_2 \cdots i_{m-1}}, 
        \quad \forall i \in \mathcal{V}.
    \end{equation}

    \item \textbf{Rate of Convergence:} The asymptotic convergence rate is determined by the modulus of the second largest U-eigenvalue of $\widetilde{P}$.
\end{enumerate}
\end{thm}

This establishes both the existence and uniqueness of the steady-state behavior for Markov chains with memory in the unified tensor framework and is an extension to the result of  \cite[Lemma 3.2]{wu2017markov}. Firstly, \cite[Lemma 3.2]{wu2017markov} is a special case of $m=2$. Secondly, the condition is not directly defined on the tensor $\widetilde{P}$ but a constructed matrix related to $\widetilde{P}$, which restricts its applicability compared with our theorem.
\hjk{
\begin{remark}\label{rem:2}
    In  \cite{wu2017markov}, a mean-field closure $\Pi \approx x^{\otimes(m-1)}=x\otimes x \otimes \cdots$ was introduced for the discrete-time case, reducing the update rule \eqref{eq:pit} to a nonlinear iteration $z^+=P\,z^{m-1}$. The steady-state equation $z^*=P\,(z^*)^{m-1}$ is then a Z-eigenvalue problem with eigenvalue $1$, and the computational cost per step drops from $O(n^{2(m-1)})$ to $O(n^m)$.  We develop the continuous-time analogue of this approximation in Section \ref{subsec:continuous_time}.
\end{remark}
}

% \begin{remark}\label{rem:2}
% Wu and Chu \cite{wu2017markov} further analyzed discrete-time higher-order Markov chains by introducing a mean-field closure, where the joint memory distribution is approximated as a rank-one Kronecker power of the state vector, i.e., 
% $\Pi \approx x^{\otimes(m-1)}=x\otimes x \otimes \cdots$, $\otimes$ is the Kronecker product (slightly modified tensor version, see \cite{wu2017markov}). 
% Substituting this into the discrete update rule
% $\Pi_{t+1} = P \ast \Pi_t,$
% where $P$ is the transition tensor, collapses the high-dimensional linear system to a nonlinear discrete-time system for the marginal vector $x$:
% %\begin{equation}\label{eq:z_eig_eq}
% $z^+=P\,z^{m-1},$ and the steady-state equation is then $z^*=P\,(z^*)^{m-1}$
% %\end{equation}
% which is the \emph{Z-eigenvalue problem} with Z-eigenvalue $1$. 
% This method for computing steady distribution avoids explicitly forming the $n^{m-1}\!\times n^{m-1}$ unfolded matrix, reducing the computational cost per step from $O(n^{2(m-1)})$ to $O(n^m)$, which is particularly efficient when $m$ is small and $n$ is large. 
% However, it only yields the marginal stationary distribution $\widetilde{x}$ rather than the full joint distribution $\widetilde{\Pi}$, and the intermediate updates of $z$ do not represent the actual trajectory of $x(t)$. 
% \end{remark}

\subsection{Continuous-Time Markov Chains with Memory}
\label{subsec:continuous_time}

In this subsection, we extend our tensor-based formulation to continuous-time Markov chains with memory, whose asymptotic behavior is rarely studied in literature but will be discovered in this subsection.

\subsubsection{Flow Rate Tensor}

Consider a continuous-time process $\{X_t\}_{t\ge 0}$ on $\mathcal{V}=\{1,\dots,n\}$ with memory depth $m$.
We introduce a nonnegative \emph{inflow rate tensor}
\[
\mathcal{R}=[r_{i_1 i_2 \cdots i_m}] \in \mathbb{R}_{\ge 0}^{n \times \cdots \times n},
\]
where $r_{i_1 i_2 \cdots i_m}$ is the instantaneous rate of moving from the ordered history
$(i_2,\ldots,i_m)$ to the new state $i_1$.
For each fixed history $(i_2,\ldots,i_m)$, define the total outflow rate
\[
\rho_{i_2 \cdots i_m} \;:=\; \sum_{j=1}^n r_{j\, i_2 \cdots i_m}\,.
\]

We then build two even-order paired tensors:

\paragraph{ Paired inflow operator.}
\[
\widetilde{\mathcal{R}}_{\,i_1 j_1\, i_2 j_2\, \cdots\, i_{m-1} j_{m-1}}
=
\begin{cases}
r_{i_1 i_2 \cdots i_m}, & j_1=i_2,\; j_2=i_3,,\\
& \; \ldots,\; j_{m-1}=i_m \\
0, & \text{otherwise.}
\end{cases}
\]

\paragraph{ Diagonal outflow operator} All entries are zeros except, 
$\mathcal{D}_{\,j_1 j_1\, j_2 j_2\, \cdots\, j_{m-1} j_{m-1}}
\;=\;
\rho_{\,j_1 \cdots j_{m-1}}\,.$

The continuous-time rate tensor on the paired space is then defined by
$\;\widetilde{\mathcal{Q}} \;:=\; \widetilde{\mathcal{R}} \;-\; \mathcal{D}\;,$
which guarantees probability conservation (column sums zero) after unfolding.

\subsubsection{Kolmogorov Forward Equation}

Let $\Pi(t)$ be the joint probability tensor of the most recent $m-1$ states.
Its time evolution satisfies
\begin{equation}\label{eq:master_equation}
\frac{d}{dt}\,\Pi(t) \;=\; \widetilde{\mathcal{Q}} \ast \Pi(t)
\;=\; \widetilde{\mathcal{R}} \ast \Pi(t) \;-\; \mathcal{D}  \ast \Pi(t).
\end{equation}
Under unfolding, \eqref{eq:master_equation} becomes the classical continuous-time Markov chain \cite{norris1998markov}: \hjk{$\frac{d}{dt}\,\varphi(\Pi(t))
\;=\;
\left(\,\varphi(\widetilde{\mathcal{R}})\;-\;\varphi({\mathcal{D}})\,\right)\;\varphi(\Pi(t))$.}
% \[
% \frac{d}{dt}\,\varphi(\Pi(t))
% \;=\;
% \big(\,\varphi(\widetilde{\mathcal{R}})\;-\;\varphi({\mathcal{D}})\,\big)\;\varphi(\Pi(t)).
% \]

\subsubsection{Stationary Behavior and Transients}

If $\widetilde{\mathcal{Q}}$ is U-irreducible, there exists a unique stationary joint distribution
$\widetilde{\Pi}$ satisfying
\[
\widetilde{\mathcal{Q}} \ast \widetilde{\Pi} \;=\; 0,
\qquad
\tilde{x}_i
=\sum_{i_2,\ldots,i_{m-1}} \widetilde{\Pi}_{i\, i_2 \cdots i_{m-1}},
\ \forall i\in\mathcal{V}.
\]
The spectrum of $\widetilde{\mathcal{Q}}$ governs the dynamics:
 the U-eigenvector associated with the zero U-eigenvalue corresponds to $\widetilde{\Pi}$, while all other U-eigenvalues have
negative real parts and determine the decay rates.

For continuous-time systems, similarly to the idea introduced in Remark \ref{rem:2}, we apply the same approximation 
$\Pi \approx x^{\otimes(m-1)}$ into the equation \eqref{eq:master_equation}
leading to
\begin{equation}\label{eq:lap}
    \dot{x} = \mathcal{R} x^{m-1} - F x^{m-1},
\end{equation}
where $F$ has the same dimension as $\mathcal{R}$ and all entries are zero except 
$F_{i i i_2\ldots i_{m-1}} = \rho_{i i_2 \cdots i_{m-1}}$. 
Define the Laplacian $L = F - \mathcal{R}$. When $n$ is sufficiently large, the assumption $\Pi \approx x^{\otimes(m-1)}$ becomes asymptotically exact, analogous to the approximation techniques used in \cite{cui2025sis}, and corresponds to the classical \emph{propagation of chaos} phenomenon \cite{sznitman2006topics}. \hma{This closure neglects the covariances among the states and is therefore based on an independence assumption between different components.}
At steady state, it yields 
$L\,(x^*)^{m-1} = \mathbf{0}$, i.e., determining whether zero is an \emph{H-eigenvalue} of $L$. 
For continuous-time processes, the approximation naturally yields 
an H-eigenvalue problem, while the Z-eigenvalue formulation is specific to the discrete-time setting.
\hjk{
\begin{remark}
    This reduces a linear system of dimension $n^{m-1}$ to a nonlinear system of dimension $n$. Note that the stationary distribution is an H-eigenvector of $L$ with a zero H-eigenvalue, similar to the classical case.
\end{remark}
}

% \begin{remark}
% This is a considerable model reduction technique for a large $m$ given a large $n$. It starts from a high-dimensional linear Markov process with memory (dimension $n^{(m-1)}$) and approximates 
% it by a low-dimensional nonlinear system (dimension $n$). It is also worth mentioning that the stationary distribution is an H-eigenvector associated with a zero H-eigenvalue of $L$, similar to the classical case.
% \end{remark}

In the following, we provide a detailed analysis of \eqref{eq:lap}.

\begin{lemma}[Positivity]\label{thm:pos}
    The system \eqref{eq:lap} is a positive system, i.e. if $x(0) \geq 0$, then $x(t) \geq 0$ for all $t \geq 0$. Furthermore, $\mathbf{1}^{\top} x(0)=M$ is a \hjk{conserved} quantity of the system.
\end{lemma}

\begin{proof}
    At the boundary, when $x_i=0$, we have $\dot{x}_i\geq 0$ due to the structure of $F$, \hjk{showing the first statement}. \hjk{In addition},
    $\frac{d}{d t} \mathbf{1}^{\top} x=\mathbf{1}^{\top}\left(\mathcal{R} x^{m-1}-F x^{m-1}\right)=0$\hjk{; showing the second.} 
    %Thus, $\mathbf{1}^{\top} x(0)=M$ is a conservative quantity.
\end{proof}

As mentioned, $L\,(x^*)^{m-1} = \mathbf{0}$ yields equilibra of \eqref{eq:lap}. Next, we further focus on a special case, and we are able to give further analytical results.

\begin{lemma}[Detailed-balance equilibrium]
    Consider the system \eqref{eq:lap}. Let $I=\{i_3,\cdots, i_m\}$. If there exists $x^{*}>\mathbf{0}$ such that the tensor $\mathcal{R}$ satisfies
\begin{equation}\label{eq:xstar}
    r_{i_1 i_2 I} x_{i_2}^{*}=r_{i_2 i_1 I} x_{i_1}^{*}, \quad \forall i_1, i_2, I ;
\end{equation}
then $L\,(x^*)^{m-1} = \mathbf{0}$. Thus, $x^{*}$ is a positive equilibrium of \eqref{eq:lap}.
\end{lemma}

\begin{proof}
    By definition,
$
\left(\mathcal{R}\left(x^{*}\right)^{m-1}\right)_{i_1}=\sum_{i_2, I} r_{i_1 i_2 I} x_{i_2}^{*} \prod_{k=3}^m x_{i_k}^{*} .
$ Using $\rho_{i_1 I}=\sum_j r_{j i_1 I}$ and the form of $F$,
$
\left(F\left(x^{*}\right)^{m-1}\right)_{i_1}=\sum_{i_2, I} r_{i_2 i_1 I} x_{i_1}^{*} \prod_{k=3}^m x_{i_k}^{*}.
$ 
Thus, it directly yields $\left(\mathcal{R} x^{m-1}-F x^{m-1}\right)=\mathbf{0}$.
\end{proof}

\begin{remark} 
A direct observation is that a supersymmetric $\mathcal{R}$ directly implies \eqref{eq:xstar} with $x^*=\mathbf{1}$. \hjk{In fact,}  \eqref{eq:xstar} represents a higher-order detailed-balance relation: 
for every pair of interactions $(i_1,i_2,I)$, the inflow and outflow fluxes are equal at the equilibrium $x^{*}$\hjk{; hence enabling convergence to the equilibrium, see Theorem \ref{thm:global}.} 
\end{remark}

Next, we introduce the concept of \emph{interaction graph}, which can be considered as a kind of \emph{projected graph} of a hypergraph. Define a directed graph
$
\mathcal{G}_R=(\mathcal{V}, \mathcal{E})
$
where the vertex set is $\mathcal{V}=\{1, \ldots, n\}$;
and the edge set:
$
\mathcal{E}:=\left\{\left(i_2, i_1\right) \mid \exists I=\left(i_3, \ldots, i_m\right) \text { s.t. } r_{i_1 i_2 I} >0\right\} .
$ Equivalently, there is a directed edge $i_2 \rightarrow i_1$ if node $i_1$ can directly receive positive inflow from node $i_2$ for at least one configuration of the other indices. Now, we can characterize the convergence behavior of the system \eqref{eq:lap}.

\begin{thm}[Global convergence]\label{thm:global}
    Consider the system \eqref{eq:lap}. If there exists an $x^{*}>\mathbf{0}$ satisfying \eqref{eq:xstar} and the corresponding $\mathcal{G}_R$ based on the tensor $\mathcal{R}$ is strongly connected, then the trajectory from $x(0)$ converges to $\alpha x^{*}>\mathbf{0}$ with $\alpha=\frac{M}{\mathbf{1}^{\top} {x}^{*}}$ and $M=\mathbf{1}^{\top} x(0)$.
\end{thm}

\begin{proof}
    First, we need to define a useful notation of ``flux": $\Phi_{i_1 i_2 I}(x):=r_{i_1 i_2 I} x_{i_2} \prod_{k=3}^m x_{i_k} \geq 0$. Then, component-wise, \eqref{eq:lap} becomes 
    % \hjk{$\dot{x}_{i_1}=\sum_{i_2, I} \Phi_{i_1 i_2 I} (x)-\sum_{i_2, I} \Phi_{i_2 i_1 I} (x).$}

    \begin{equation}
\dot{x}_{i_1}=\sum_{i_2, I} \Phi_{i_1 i_2 I} (x)-\sum_{i_2, I} \Phi_{i_2 i_1 I} (x).
\end{equation}
Define a Lyapunov function \hjk{$V(x)=\sum_{i=1}^n\left[x_i \ln \frac{x_i}{x_i^{*}}-x_i+x_i^{*}\right] \geq 0$.}
% $$V(x)=\sum_{i=1}^n\left[x_i \ln \frac{x_i}{x_i^{*}}-x_i+x_i^{*}\right] \geq 0.$$  
It follows that
\begin{equation}\label{eq:V}
    \begin{split}
        \dot{V}&=\sum_i \ln \frac{x_i}{x_i^{*}} \dot{x}_i=\sum_{i_1, i_2, I}\left(\ln \frac{x_{i_1}}{x_{i_1}^{*}}-\ln \frac{x_{i_2}}{x_{i_2}^{*}}\right) \Phi_{i_1 i_2 I} \\
    \end{split}
\end{equation}
Note that the two terms above counts both ordered pairs 
$(i_1,i_2)$ and their reversed pairs $(i_2,i_1)$. 
To make the expression symmetric, we add and subtract the same term 
with exchanged indices and take the average:
$\sum_{i_1,i_2,I}T_{i_1 i_2}\Phi_{i_1 i_2 I}
=\tfrac12\sum_{i_1,i_2,I}\!\big(T_{i_1 i_2}\Phi_{i_1 i_2 I}
   +T_{i_2 i_1}\Phi_{i_2 i_1 I}\big),$
where $T_{i_1 i_2}:=\ln\!\frac{x_{i_1}}{x_{i_1}^{*}}
-\ln\!\frac{x_{i_2}}{x_{i_2}^{*}}$ and thus $T_{i_2 i_1}=-T_{i_1 i_2}$. This yields that $\dot{V}=\frac{1}{2} \sum_{i_1, i_2, I}\left(\ln \frac{x_{i_1}}{x_{i_1}^{*}}-\ln \frac{x_{i_2}}{x_{i_2}^{*}}\right)\left(\Phi_{i_1 i_2 I}-\Phi_{i_2 i_1 I}\right).$

Next, let $
\sigma_{i_1 i_2 I}:=r_{i_1 i_2 I} x_{i_2}^{*}=r_{i_2 i_1 I} x_{i_1}^{*} \geq 0, \quad C_I:=\prod_{k=3}^m x_{i_k}^{*}>0,
$
and define
$
a:=\frac{x_{i_2}}{x_{i_2}^{*}} \prod_{k=3}^m \frac{x_{i_k}}{x_{i_k}^{*}}, \quad b:=\frac{x_{i_1}}{x_{i_1}^{*}} \prod_{k=3}^m \frac{x_{i_k}}{x_{i_k}^{*}} .
$
Then, we have 
$
\Phi_{i_1 i_2 I}-\Phi_{i_2 i_1 I}=\sigma_{i_1 i_2 I} C_I(a-b), \quad \ln \frac{x_{i_1}}{x_{i_1}^{*}}-\ln \frac{x_{i_2}}{x_{i_2}^{*}}=\ln \frac{b}{a}.
$
Plugging the above into \eqref{eq:V} yields:
$
\dot{V}=\frac{1}{2} \sum_{i_1, i_2, I} \sigma_{i_1 i_2 I} C_I(a-b) \ln \frac{b}{a} .
$

For any $u, v>0$, one has $(u-v) \ln (u / v) \geq 0$. 
%This could be proven by Mean Value Theorem and is omitted here. 
With $u=b, v=a$, we have 
$
(a-b) \ln \frac{b}{a}=-(b-a) \ln \frac{b}{a} \leq 0.
$  Since $\sigma_{i_1 i_2 I} C_I \geq 0$, we conclude that $\dot{V}\leq 0$, which holds as an equality iff $a=b$. Thus $\dot{V} \equiv 0$ iff for all $i_1, i_2, I$ with $\sigma_{i_1 i_2 I}>0$,
$
\frac{x_{i_1}}{x_{i_1}^{*}}=\frac{x_{i_2}}{x_{i_2}^{*}} .
$
By strong connectivity, this forces
$
\frac{x_1}{x_1^{*}}=\cdots=\frac{x_n}{x_n^{*}}=: \alpha.
$
LaSalle's invariance principle with $\dot{V} \leq 0$ gives
$
x(t) \rightarrow \alpha x^{*} .
$

By further considering the conservative quantity from Lemma \ref{thm:pos}, then the trajectory from $x(0)$ must converge to $\alpha x^{*}>\mathbf{0}$ with $\alpha=\frac{M}{\mathbf{1}^{\top} {x}^{*}}$.
\end{proof}

\section{Application: Random Walks on Hypergraphs}\label{sec:random_walks}
% Here, we focus on illustrating how this framework can be used to analyze and interpret 
% random walks on hypergraphs.
% We now apply the higher-order Markov chain framework developed earlier to define random walks on hypergraphs.
% The normalized adjacency tensor $\Tilde{A}$ introduced in Section~2 naturally serves 
% as the transition tensor for this process.

\hjk{
We now apply the framework of Section \ref{sec:markov_memory} to define random walks on hypergraphs, using the normalized adjacency tensor $\tilde{A}$ from Section \ref{sec:preliminaries} as the transition tensor.
}

\begin{definition}[Random Walk on a Hypergraph]
Consider a hypergraph with normalized adjacency tensor $\Tilde{A}$.
A \emph{random walk with memory $m-1$} is defined as: 
at time $t$, the walker has visited vertices 
\((v_t, v_{t-1}, \ldots, v_{t-m+2})\). 
The next step proceeds in two stages:
\begin{enumerate}
    \item The walker selects a hyperedge $e$ whose ordered tail set matches 
    \(\{v_t, v_{t-1}, \ldots, v_{t-m+2}\}\),
    with probability given by the corresponding entry of $\Tilde{A}$. (If a hyperedge doesn't exist, its probability is zero. For continuous-time setting, the tensor ${A}$ denotes the inflow rate. (Normalization is unnecessary for continuous-time setting.))
    \item The walker moves to the head node $v_{t+1}=h$ of the chosen hyperedge $e$ at time $t+1$.
\end{enumerate}
This process defines a higher-order Markov chain with transition (inflow) tensor $\Tilde{A}$.
\end{definition}

\begin{figure}[t]
    \centering
    \includegraphics[width=0.8\linewidth]{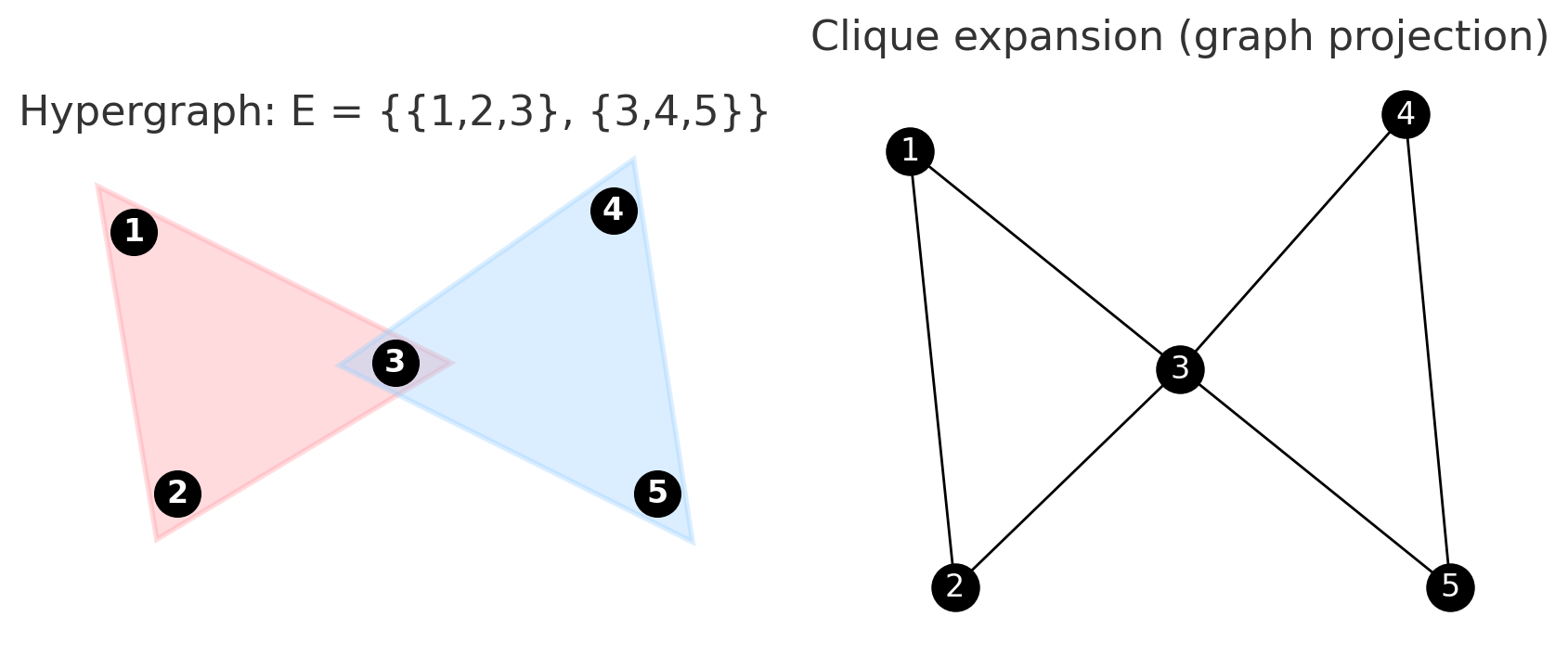}
    \caption{Left: original hypergraph $E=\{\{1,2,3\},\{3,4,5\}\}$, where all hyperedges are undirected. For example, $\{1,2,3\}$ denotes composition of all ordered tails and heads induced by  $1,2,3$. Right: the corresponding projected graph, obtained by replacing each hyperedge with a complete pairwise subgraph. All edges and hyperedges are equally weighted.}
    \label{fig:hypergraph_rw}
\end{figure}

Consider the hypergraph shown in Fig.~\ref{fig:hypergraph_rw} (left). For the \emph{memory-based random walk} on this hypergraph 
with depth $m-1=2$, each 3-hyperedge produces internal cyclic trajectories in the unfolded state space. 
As a consequence, the process splits into two closed communicating classes: one supported on $\{1,2,3\}$ 
and the other on $\{3,4,5\}$. Using the tensor unfolding method introduced earlier \eqref{eq:master_equation}, one can compute 
the corresponding stationary node distributions as
$
x^{*(1)} = \Big(\tfrac{1}{3}, \tfrac{1}{3}, \tfrac{1}{3}, 0, 0\Big), 
\qquad
x^{*(2)} = \Big(0, 0, \tfrac{1}{3}, \tfrac{1}{3}, \tfrac{1}{3}\Big).
$
Hence, the hypergraph random walk does not admit a unique global stationary distribution; the limit 
depends on the initial condition.

For comparison, Fig.~\ref{fig:hypergraph_rw} (right) depicts the corresponding \emph{projected graph} \cite{carletti2020random,carletti2021random}, in which 
each hyperedge is replaced by a complete pairwise subgraph among each hyperedge. On this graph, a classical random walk 
is irreducible and aperiodic, and thus possesses a unique stationary distribution. The resulting distribution is
$x^* = \Big(\tfrac{1}{6}, \tfrac{1}{6}, \tfrac{1}{3}, \tfrac{1}{6}, \tfrac{1}{6}\Big),$
where node~3 obtains the largest weight.

\hjk{
This example reveals a fundamental difference from classical approaches. Most existing hypergraph random walks \cite{carletti2020random,carletti2021random} can be reformulated as memoryless walks on a weighted projected graph, since each hyperedge is reduced to pairwise connections. In contrast, our formulation is built on a higher-order Markov chain with memory, where the ordered tail of each hyperedge encodes past states. This process cannot be collapsed into a projection graph; instead, it corresponds to a walk on an unfolded graph whose nodes represent memory sequences of length $m-1$, thereby revealing behavior inaccessible to memoryless approaches.
}

% This example shows a fundamental difference between our definition of random walks on hypergraphs and the classical random walk. Furthermore, most existing definitions of hypergraph random walks can ultimately be reformulated as standard 
% random walks on a weighted projected graph, where each hyperedge is reduced to pairwise 
% connections with appropriate weights. 
% For instance, Carletti \emph{et al.} \cite{carletti2020random,carletti2021random} consider 
% \emph{uniform diffusion}, where the walker first selects uniformly among hyperedges containing the 
% current node and then moves uniformly to another node within that hyperedge. 
% This process is mathematically equivalent to a memoryless random walk on a weighted graph whose edge 
% weights reflect hyperedge projection. 
% In contrast, our formulation is explicitly built on a higher-order Markov chain with memory, where 
% the ordered tail of each hyperedge encodes a sequence of past states. 
% While this process cannot be collapsed into a simple weighted projection graph, it can be interpreted 
% as a random walk on an unfolded extended graph whose nodes represent all possible memory sequences of length $m-1$. 
% Consequently, our model reveals higher-order random walk behavior that is fundamentally inaccessible to 
% memoryless approaches based solely on weighted projection graphs.

\begin{figure}[t]
    \centering
    \includegraphics[width=0.8\linewidth]{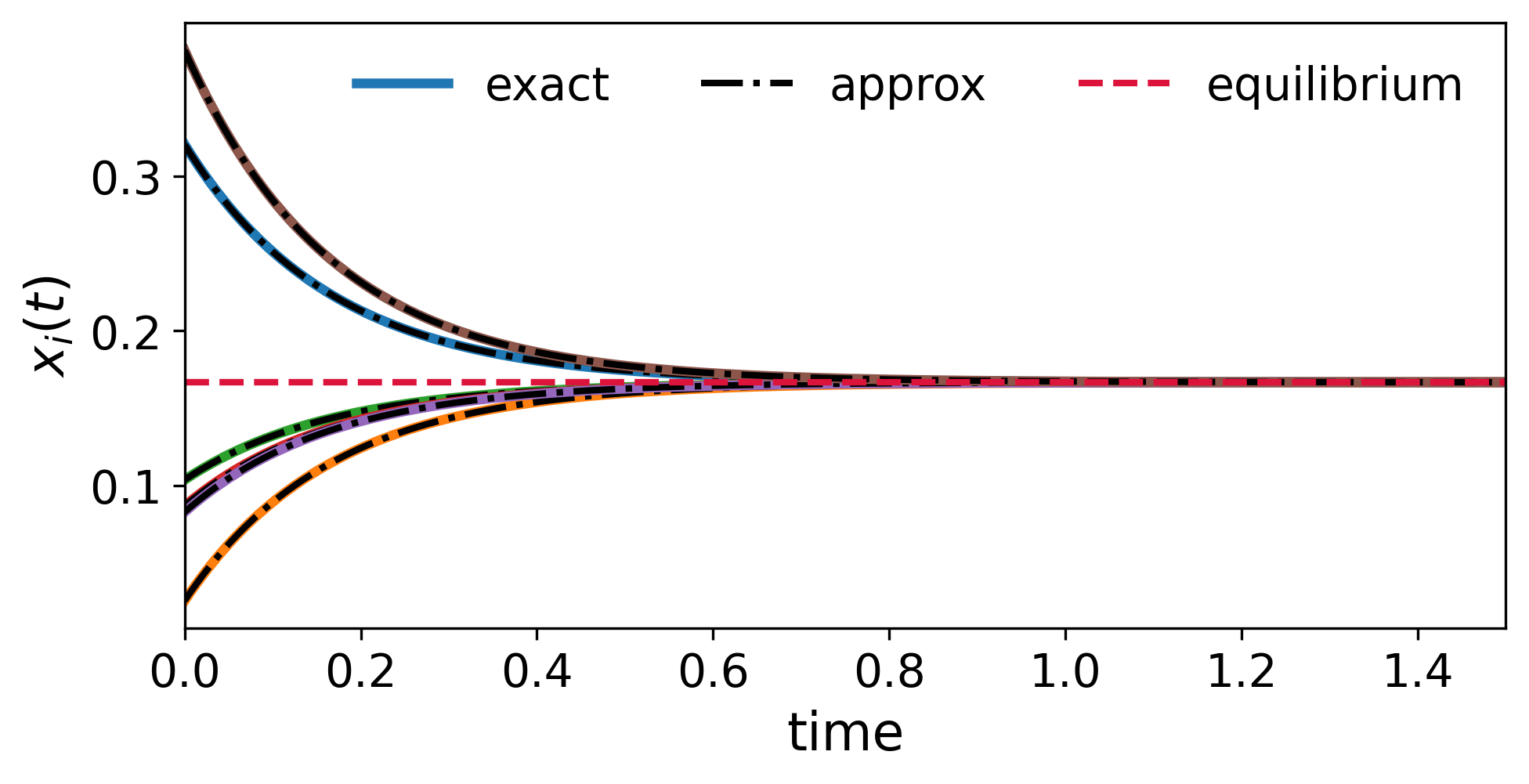}
    \caption{
Comparison between the unfolded higher-order Markov dynamics \eqref{eq:master_equation} and its nonlinear Laplacian approximation \eqref{eq:lap}. 
Both models start from the same initial marginal $x(0)$ with $\mathbf{1}^{\top} x(0)=1$ and $\Pi(0)=x(0)^{\otimes (m-1)}$ in the unfolded system. 
In the supersymmetric (undirected) case, their trajectories are sufficiently close and converge to the same uniform distribution $x_i = 1/n$. 
Solid lines: exact unfolded dynamics; dashed lines: nonlinear approximation; dotted line: uniform equilibrium $1/n$.
}
    \label{fig:compare}
\end{figure}

As discussed in Section~\ref{subsec:continuous_time}, the continuous-time random walk dynamics can be reduced to a higher-order nonlinear Laplacian system~\eqref{eq:lap}.
Here we perform a simulation with $n=6$ nodes and order $m=5$.
For simplicity, we set $\mathcal{R}=A$ as an all-one supersymmetric tensor.
We compare the simulation results obtained from the exact system~\eqref{eq:master_equation} and from its nonlinear Laplacian approximation~\eqref{eq:lap}, as illustrated in Fig.~\ref{fig:compare}.
The results show that the higher-order nonlinear Laplacian model~\eqref{eq:lap} closely matches the exact continuous-time random walk dynamics in both transient evolution and steady-state distribution.

\section{Conclusion}

\hjk{
This paper develops a unified tensor framework for higher-order Markov chains with memory and random walks on hypergraphs. The even-order paired tensor representation links folded and unfolded dynamics, and the nonlinear Laplacian reduction enables a full stability analysis under a higher-order detailed-balance condition. Future directions include multi-player settings with strategic updates, heterogeneous memory depths on non-uniform hypergraphs, and the analysis of multiple equilibria in the nonlinear Laplacian dynamics \eqref{eq:lap}.
}

% This paper develops a unified tensor-based framework for modeling higher-order Markov chains with memory
% and demonstrates its application to random walks on hypergraphs. 
% By introducing an even-order paired tensor representation, we establish a direct connection between 
% transition dynamics, eigenvalue properties, and stationary behavior. 
% Our formulation is more general than existing approaches and reveals the precise mathematical relationship 
% between folded and unfolded representations. Then, we introduce the nonlinear Laplacian system analysis for higher-order random walk dynamics. This nonlinear formulation provides an interpretable reduced model that approximates both transient and steady-state behaviors of the original Markov system. Under a mild condition, we give a full analysis of the system behavior of such a nonlinear Laplacian system.
% For future work, we will extend the framework to \emph{multi-player} settings on hypergraphs, coupling interacting agents/walkers with strategic updates. 
% We will also consider integrating heterogeneous memory depths with \emph{non-uniform} hypergraphs. For the nonlinear Laplacian dynamics \eqref{eq:lap}, it is also worthwhile studying the case when there are multiple equilibra. 

\bibliographystyle{IEEEtran}
\bibliography{bib}

\end{document}